\documentclass[runningheads,a4paper]{llncs}
\usepackage{amsmath}
\usepackage{amssymb}
\usepackage{amsfonts}
\usepackage{subfigure}
\usepackage{graphicx}
\usepackage[british]{babel}

\usepackage{setspace}
\newcommand{\cG}{{\mathcal G}}
\newcommand{\cX}{{\mathcal X}}

\newcommand{\ba}{\backslash}

\begin{document}

\title{Quantifying the Extent of Lateral Gene Transfer Required to Avert a `Genome of Eden'}
\titlerunning{Quantifying LGT}
\author{Leo van Iersel, Charles Semple and Mike Steel \thanks{We thank the Allan Wilson Centre for Molecular Ecology and Evolution, and the New Zealand Marsden Fund for helping fund this work.}
\institute{Department of Mathematics and Statistics, University of Canterbury\\Private Bag 4800, Christchurch, New Zealand}\\
}

\maketitle

\begin{abstract}
The complex pattern of presence and absence of many genes across different species provides tantalising clues as to how genes evolved through the processes of gene genesis, gene loss and lateral gene transfer (LGT).  The extent of LGT, particularly in prokaryotes, and its implications for creating a `network of life' rather than a `tree of life' is controversial.
In this paper, we
formally model the problem of quantifying LGT, and provide exact mathematical bounds, and new computational results.  In particular, we investigate the computational complexity of quantifying the extent of
LGT under the simple models of gene genesis, loss and transfer on which a recent heuristic analysis of biological data relied.
Our approach takes advantage of a relationship between LGT optimization and graph-theoretical concepts such as tree width and network flow.\\
\end{abstract}


{\bf Email:} l.j.j.v.iersel@gmail.com, c.semple@math.canterbury.ac.nz, m.steel@math.canterbury.ac.nz
\newpage

\section{INTRODUCTION}

Modern sequencing technology is providing an increasingly detailed picture of the distribution of genes across a wide array of taxa.
Some molecular biologists have used these data to argue that unless ancestral genomes were considerably larger than present-day ones,
extensive lateral gene transfer (LGT) must be invoked to explain the current distribution of genes \cite{dag07}, \cite{dag08}, \cite{mir03}.
LGT is a process by which a gene (or genes) from one species is transferred into the genotype of another species by various genetic mechanisms.  The extent of LGT is controversial, but it has been argued to be widespread in prokaryotes (e.g. bacteria) and during the earlier epochs of evolution, suggesting in turn that a network, rather than a tree, best describes the evolution of life~\cite{doo07}.

Although the pattern of  presence and absence of different genes across a set of species can suggest that LGT events occurred in the evolution of these species, another explanation is that certain genes are simply lost in different lineages. As a result, various attempts to quantify the extent of LGT based on gene content have been developed, typically based either on most-parsimonious scenarios or on stochastic models of gene genesis, loss and transfer (see, for example,  \cite{dag07}, \cite{jin07}, \cite{spe06}).   Attempts  to reconstruct evolutionary histories under the assumption that no LGT  events have occurred (and that genes arise just once) imply that some common ancestors of the considered species must have had far  more genes than their current-day descendants.  Doolittle {\em et al.} \cite{doo03} refer to such an unlikely all-encompassing ancestral genome as the `genome of Eden' hypothesis.   Allowing LGT events reduces the need for genes to be present at earlier species, as illustrated for a single gene in Fig.~\ref{intro}.

\begin{figure}[ht]\centering
\resizebox{10.0cm}{!}{
\includegraphics[width=\textwidth]{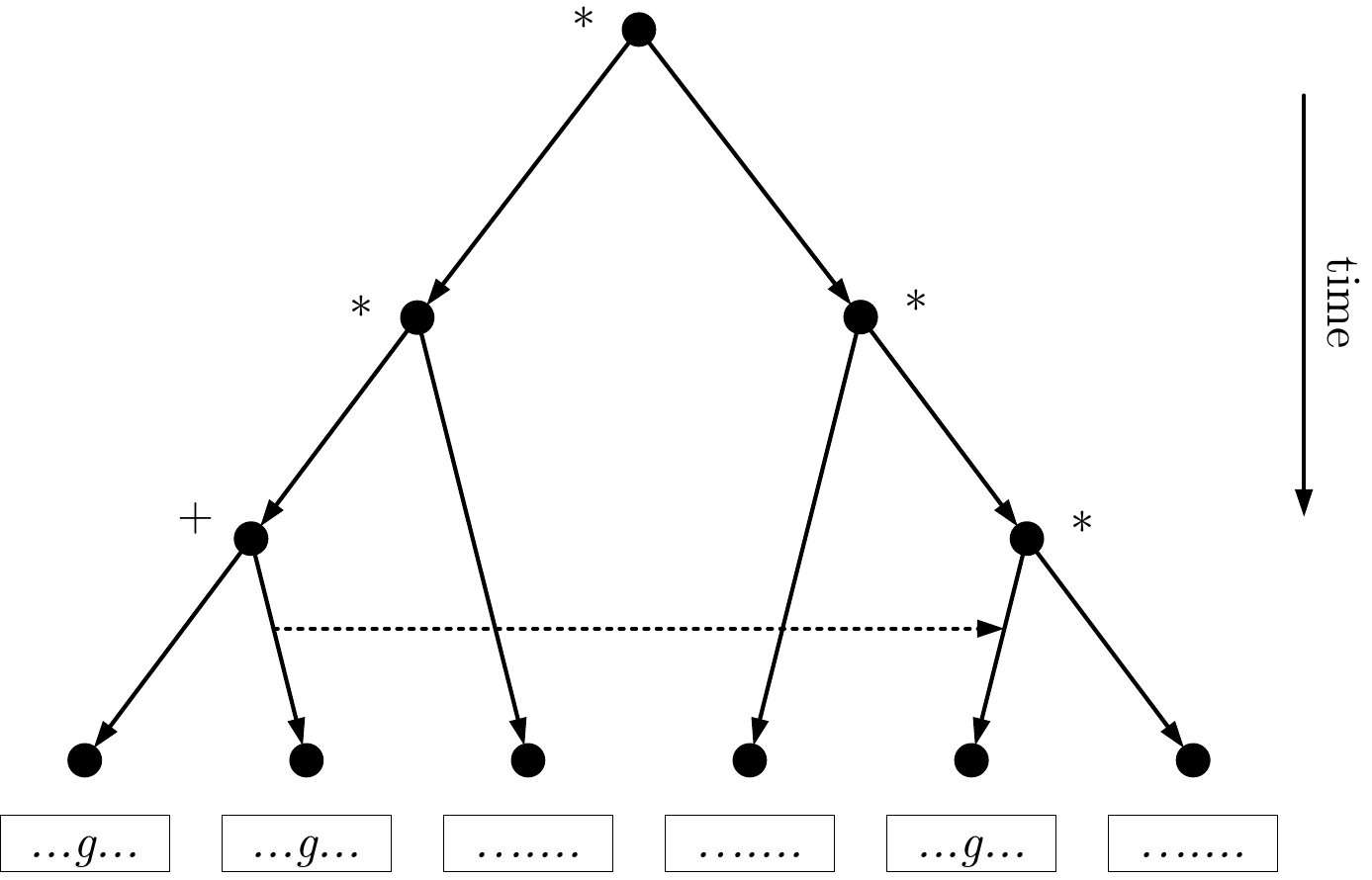}
}
\caption{ The dilemma of ancestral genome inflation: If gene $g$, distributed as shown, is not  transferred laterally  then under the model, $g$ must be in five ancestral genomes (*,+) not just at +. }
\label{intro}
\end{figure}

In this paper, we exploit the combinatorial  structure that underlies a key biological insight on which a recent heuristics analysis of data was based by \cite{dag07} (see also \cite{dag08}, \cite{mir03}). This insight is that simple models of gene evolution, in which a gene typically arises just once (gene genesis) but can be lost multiple times, imply lower bounds on the extent of LGT simply to prevent hypothetical ancestral genomes from becoming unfeasibly large.  For such a model, we aim to bound the number of gene transfer events that have occurred in the evolution of a set of taxa, based on the presence/absence patterns of genes in each of these taxa, assuming that ancestral genomes are bounded by a given size.

Notice that we wish to count transfer events  (rather than the total number of genes that are transferred), since in each transfer event,  several genes may be transferred from one species into another. Thus our count of LGTs is conservative, and recognizes that genes are not independently transferred and that a transfer event may insert a section of the genome (with several genes) into an individual organism of a  different species.

The structure of this paper is as follows. In the next section, we define the model of gene genesis, loss and transfer  precisely, and summarize our main results. We then provide proofs of these results in subsequent sections, and end with some concluding comments and a conjecture.

\section{MATHEMATICAL MODEL AND SUMMARY OF MAIN RESULTS}

\subsection{Definitions and model specification}

We begin by recalling some notation concerning digraphs, and phylogenetic trees and networks.

Let $v$ be a vertex of a digraph $D$. The {\em indegree} of $v$ is the number of arcs directed into
$v$, while the {\em outdegree} of $v$ is the number arcs directed out of $v$. The indegree of $v$ is denoted by $d^-(v)$ and the outdegree of $v$ is denoted by $d^+(v)$. The \emph{degree} of~$v$ is~$d^-(v)+d^+(v)$. Furthermore, $u$ is
an {\em in-neighbour} of $v$ if $(u,v)$ is an arc in $D$, while $w$ is an {\em out-neighbour}
of $v$ if $(v,w)$ is an arc in $D$. A digraph $D$ is {\em rooted} if there exists a vertex, $\rho$ say, of indegree zero such that, for each vertex $v$ in $D$, there exists a directed path from $\rho$ to $v$.

Throughout the paper, $\cX$ will denote a finite set of taxa and $\cG$ will denote a finite set of genes. A {\em phylogenetic tree (on $\cX$)} is a rooted tree whose root has degree at least two and all other internal vertices have degree at least three, and whose leaf set is
$\cX$. More generally, a
{\em phylogenetic network $N$ (on $\cX$)} is a rooted acyclic digraph with the following properties:
\begin{itemize}
\item[(i)] the root has outdegree at least two and, for all vertices $v$ with $d^+(v)=1$, we have
$d^-(v)\ge 2$; and

\item[(ii)] the set of vertices of outdegree zero is $\cX$.
\end{itemize}
The elements of $\cX$ are the {\em leaves} of $N$. For a subset $U$ of the vertex set of $N$, the sub-digraph of $N=(V,A)$ {\em induced by $U$} is the digraph whose vertex set is $U$, and whose arc set is the subset $\{(u,v): \mbox{$u,v\in U$ and $(u,v)\in A$}\}$ of $A$.

We now describe the model of gene genesis, loss, and transfer. For each taxon~$x\in\cX$, assume that the subset $G(x)$ of $\cG$ consisting of the genes in $\cG$ that have been observed in taxon $x$ is known.
We refer to the associated map $G: \cX \rightarrow 2^\cG$ as a {\em genome assignment}.  Let~$N=(V,A)$ be a phylogenetic network on $\cX$. For a fixed positive integer $k$, and a genome assignment
$G:  \cX \rightarrow 2^\cG$, a {\em $(G,k)$-gene labelling of $N$} is a mapping $F:V\rightarrow 2^{\cG}$ such that the following hold:
\begin{itemize}
\item[(I)] $F(x)=G(x)$ for each $x\in\cX$;
\item[(II)] $|F(v)|\le k$ for all $v\in V$;
\item[(III)] For each gene $g\in\cG$, the sub-digraph of $N$ induced by $\{v\in V : g\in F(v)\}$ is rooted (and therefore connected).
\end{itemize}
Note that if $x\in \cX$ and $|G(x)|>k$, then $N$ has no $(G,k)$-labelling. If $N$ has a $(G,k)$-labelling, we say that $N$ {\em exhibits} such a labelling. A gene labelling describes a possible evolution of the genes observed in the
taxa under consideration.  Property (I) says that each leaf of the network is labelled by the set of genes observed in the corresponding taxon. Property (II) demands that each vertex is labelled by a set of at most~$k$ genes;  the parameter $k$ thus bounds the sizes of the ancestral genomes. Lastly, (III), means that each gene in $\cG$ is created once at most. There is no restriction on the number of times a gene is lost.

Any function $F$ which satisfies properties $(I)$ and $(III)$ we will call a {\em $G$-gene labelling}. With these definitions in hand we can now state the main results of this paper.

\subsection{Bounding the number of gene transfers required}

Our first result establishes lower and upper bounds on the number of LGT events required to explain a given data set. Suppose our input is given by a rooted phylogenetic tree $T$ on $\cX$ (``species tree''), a genome assignment $G: \cX \rightarrow 2^\cG$, and a positive integer $k$. Given a phylogenetic network~$N$, we say that~$N$ can be obtained from~$T$ by \emph{adding}~$h$ \emph{arcs}, if there is a subgraph~$T'$ of~$N$ that is a subdivision of~$T$ (i.e.~$T'$ can be obtained from~$T$ by replacing arcs by directed paths) and~$h$ arcs of~$N$ are not arcs of~$T'$. Here, one views these added arcs as LGT events.

We are interested in the minimum number of LGT events that must be added to $T$ in order for the resulting network to exhibit a $(G,k)$-gene labelling. We denote this minimum number by $\ell(T,G,k)$. Given the above input, Theorem~\ref{prop:LB} provides lower and upper bounds for $\ell(T,G,k)$. For a vertex $v$ of $T$, let $n(v)$ denote the number of genes $g\in\cG$ for which there exist two leaves $x_1,x_2\in\cX$ such that $g\in G(x_1)$, $g\in G(x_2)$ and the most recent common ancestor of~$x_1$ and~$x_2$ in~$T$ is~$v$.

\begin{theorem}\label{prop:LB}
Let $T=(V,E)$ be a rooted phylogenetic tree on $\mathcal{X}$, let $\cG$ be a set of genes, let $G: \cX \rightarrow 2^\cG$ be a genome assignment,
and let $k$ be a positive integer. Then:
\begin{itemize}
\item[{\rm (i)}] $\ell(T,G,k) \geq \sqrt{\frac{2}{3}|\{v\in V : n(v)>k\}|}$.

\item[{\rm (ii)}] $\ell(T,G,k) \leq \Big\lceil\frac{|\mathcal{G}|-k}{k}\Big\rceil \cdot (|\mathcal{X}|+1)$.
\end{itemize}
\end{theorem}
The proof of Theorem~\ref{prop:LB} is given in Section~\ref{boundproof}.

\subsection{Hardness results}

The next two results show that two fundamental decision questions concerning the existence of
$(G, k$)-labellings are NP-complete. First, consider the following problem:

\noindent\begin{tabular}{lp{0.85\textwidth}}
\multicolumn{2}{l}{\textsc{Gene Labelling}} \\
\textit{Given:} & A phylogenetic network $N$ on $\cX$, a finite set $\mathcal G$ of genes, a genome assignment $G:\cX \rightarrow 2^\cG$, and a positive integer $k$. \\
\textit{Question:} & Does $N$ exhibit a $(G,k)$-labelling? \\
\end{tabular}

\begin{theorem}
\label{prop:homeomorphism}
The decision problem {\sc Gene Labelling} is NP-complete even if $k=1$.
\end{theorem}

A related problem, but concerning rooted phylogenetic trees, is the following:

\noindent\begin{tabular}{lp{0.85\textwidth}}
\multicolumn{2}{l}{\textsc{$(G,k)$-Tree}} \\
\textit{Given:} & A finite set $\mathcal X$ of taxa, a finite set $\mathcal G$ of genes,  a genome assignment $G:\cX \rightarrow 2^\cG$, and a positive integer $k$. \\
\textit{Question:} & Does there exist a rooted phylogenetic tree $N$ on $\cX$ that exhibits a
$(G,k)$-labelling?  \\
\end{tabular}

\begin{theorem}
\label{Gktree}
The decision problem {\sc $(G,k)$-Tree} is NP-complete.
\end{theorem}

The proofs of these two theorems are established Section~\ref{unravel}.

\subsection{Algorithms}

Despite the apparent intractability of the two problems described above, there are instances for which there exist polynomial-time algorithms. Several such instances are described in Section~\ref{easysec}. One in particular is given next.

Let $N$ be a phylogenetic network on $\cX$. A sequence of vertices and arcs is an {\em underlying cycle} of $N$ if it is a cycle of the underlying graph (i.e the undirected graph obtained by ignoring the directions of the arcs). A phylogenetic network $N$ on $\cX$ is a {\em galled tree} if, for each pair $C$ and $D$ of underlying cycles, the vertex sets of $C$ and $D$ are disjoint. Each such cycle is called a {\em gall}. Theorem~\ref{prop:level1} shows that restricting the phylogenetic networks in {\sc Gene Labelling} to galled trees, the  decision problem becomes polynomial-time solvable.

\begin{theorem}\label{prop:level1}
Let $N$ be a galled tree on $\cX$, let $\cG$ be a set of genes, let $G: \cX \rightarrow 2^\cG$ be a genome assignment,  and let $k$ be a positive integer. Then there is a polynomial-time algorithm for deciding whether or not $N$ exhibits a $(G,k)$-gene labelling.
\end{theorem}

\noindent Theorem~\ref{prop:level1}, together with the following corollary, is established in
Section~\ref{easysec}.

\begin{corollary}
 \label{gallcor}
Let $T$ be a rooted phylogenetic tree on $\cX$, let $\cG$ be a set of genes, let $G: \cX \rightarrow 2^\cG$ be a genome assignment, and let $k$ be a positive integer. If $h$ is a fixed positive integer, then there is a polynomial-time algorithm for deciding whether or not there is a galled tree $N$ on $\cX$ that can be obtained from $T$ by adding at most $h$ arcs and which exhibits a $(G,k)$-gene labelling.
\end{corollary}

\section{HOW MANY GENE TRANSFERS ARE NEEDED?}
\label{boundproof}

In this section, we prove Theorem~\ref{prop:LB}.

\noindent{\em Proof of Theorem~\ref{prop:LB}.}
For the proof of (i), suppose that a network $N$ admitting a $(G,k)$-gene labelling can be obtained by adding $\ell(T,G,k)$ arcs to~$T$. It follows that there exists a tree~$T'$ that is a subdivision of~$T$ and a subgraph of~$N$. In other words,~$T'$ is an embedding of~$T$ in~$N$. An arc of~$N$ is said to be an \emph{lgt-arc} if it is not an arc of~$T'$. Consider two leaves~$x_1,x_2$ and their lowest common ancestor~$v$ in~$T'$. Suppose that for a gene~$g\in\mathcal{G}$ we have~$g\in G(x_1)$ and~$g\in G(x_2)$. Since network~$N$ admits a~$(G,k)$-gene labeling~$F$, there has to be an undirected path from~$x_1$ to~$x_2$ in~$N$ containing only vertices~$u$ with~$g\in F(u)$. Furthermore, at least one such undirected path has to consist of two directed paths, one ending in~$x_1$ and one ending in~$x_2$, since the subgraph of~$N$ induced by~$\{v\in V | g\in F(v)\}$ is rooted and hence contains a rooted tree. There are four possibilities. Firstly, it is possible that this undirected~$x_1-x_2$-path passes through~$v$, implying that~$g\in F(v)$.
The remaining three cases are illustrated in Fig.~\ref{fig:bound}.  The first case is that the undirected~$x_1-x_2$-path  uses an lgt-arc~$(a,d)$ between two vertices~$a,d$ that have~$v$ as their lowest common ancestor in $T'$. A second possibility is that the path uses two lgt-arcs~$(a,b)$ and~$(c,d)$ such that~$v$ is the lowest common ancestor of~$a$ and~$d$ in~$T'$. Finally, it is also possible that the path uses two lgt-arcs~$(b,a)$ and~$(c,d)$ such that~$v$ is the lowest common ancestor of~$a$ and~$d$ in~$T'$.  Thus, for any vertex~$v$ with~$n(v)>k$, there has to be either an lgt-arc~$(a,d)$ or two lgt-arcs~$(a,b),(c,d)$ or two lgt-arcs~$(b,a),(c,d)$, with~$a$ and~$d$ two vertices that have~$v$ as their lowest common ancestor in~$T'$.

Given a vertex~$v$ of~$T$, we say that an lgt-arc~$(s,t)$ \emph{satisfies}~$v$ if~$v$ is the lowest common ancestor of~$s$ and~$t$ in~$T'$. Since in a tree there is a unique lowest common ancestor, each single lgt-arc satisfies at most one vertex. Furthermore, we say that a pair of lgt-arcs~$\{(s,t),(s',t')\}$ \emph{satisfies}~$v$ if~$v$ is the lowest common ancestor of either~$s$ and~$t'$, or of~$s'$ and~$t$ or of~$t$ and~$t'$ in~$T'$. It follows directly that each pair of lgt-arcs satisfies at most three vertices. Since there are~$\ell(T,G,k)$ lgt-arcs, in total at most $3{\ell(T,G,k)\choose 2} + \ell(T,G,k)$ vertices~$v$ with~$n(v)>k$ can be satisfied. From the previous paragraph we know that each vertex~$v$ with~$n(v)>k$ needs to be satisfied, either by a single lgt-arc or by a pair of lgt-arcs. It follows that there can be at most~$3{\ell(T,G,k)\choose 2} + \ell(T,G,k)$ vertices~$v$ with~$n(v)>k$. Part (i) follows by generously bounding~$3{\ell(T,G,k)\choose 2} + \ell(T,G,k)$ by~$\frac{3}{2}\ell(T,G,k)^2$.

\begin{figure}[ht]\centering
\includegraphics[width=\textwidth]{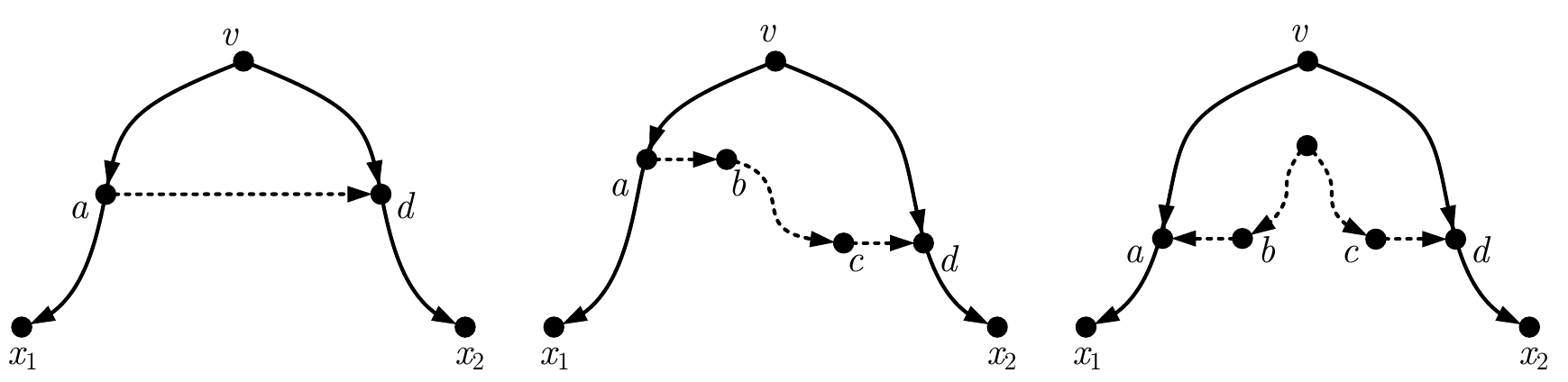}
\caption{Illustration for the proof of Theorem~\ref{prop:LB}. The three cases apply, without loss of generality, whenever ~$g\in G(x_1)$,~$g\in G(x_2)$, but~$g\not\in F(v)$, where~$v$ is the lowest common ancestor of~$x_1$ and~$x_2$ in~$T'$. Straight lines denote arcs, while curves denote paths. Solid curves are in~$T'$, while dotted lines/curves can be either in~$T'$ or only in~$N$.}\label{fig:bound}
\end{figure}

For (ii), we can construct a network~$N$ admitting a~$(G,k)$-gene labelling as follows. We select a set~$G^0$ of~$k$ arbitrary genes in~$\mathcal{G}$ and set~$F(v)=G^0$ for each internal vertex~$v$ of~$T$. The third property of a~$(G,k)$-gene labelling is now satisfied for the genes in~$G^0$. For the remaining~$|\mathcal{G}|-k$ genes we do the following. We introduce ~$f = \lceil\frac{|\mathcal{G}|-k}{k}\rceil$ additional isolated vertices~$v_1,\ldots ,v_f$ and label these vertices by disjoint sets~$F(v_1),\ldots ,F(v_f)$ that partition~$\mathcal{G}\setminus G^0$ and contain at most~$k$ genes each. Finally, we add arcs from the root to each~$v_i$ and from each~$v_i$ to each leaf~$x$ with~$G(x)\cap F(v_i)\neq\emptyset$. This leads to the claimed upper bound. \qed

To improve upon this simple upper bound turns out to be challenging. This can perhaps be explained by the results in the next section, in which we show that, even if the network~$N$ is given and~$k=1$, it is NP-complete to decide if a~$(G,k)$-gene labelling of~$N$ exists.

\section{UNRAVELLING LATERAL GENE TRANSFER IS HARD}
\label{unravel}

We begin this section by first showing that {\sc Gene Labelling} is NP-complete.  First consider the following decision problem:

\noindent\begin{tabular}{lp{0.85\textwidth}}
\multicolumn{2}{l}{\textsc{Directed Acyclic Subgraph Homeomorphism (DASH)}} \\
\textit{Given:} & Directed acyclic graphs~$D=(V_D,E_D)$ and $P=(V_P,E_P)$ with
$V_P\subseteq V_D$.\\
\textit{Question:} & Is~$P$ homeomorphic to a subgraph of~$D$? \\
\end{tabular}

A graph~$P$ is \emph{homeomorphic} to a graph~$H$ if~$H$ can be obtained from~$P$ by replacing arcs~$(u,v)$ by internally vertex-disjoint directed~$u-v$ paths. Hence, \textsc{DASH} can be seen as a disjoint-paths problem. The graph~$P$ is called the ``pattern graph''. It was observed by Fortune et al.~\cite{FortuneEtAl1980} that NP-hardness of \textsc{DASH} follows from a result of Even, Itai, and Shamir~\cite{EvenEtAl1976} on multi-commodity flows.

\noindent {\bf Theorem~\ref{prop:homeomorphism}.}
{\em The decision problem {\sc Gene Labelling} is NP-complete even if $k=1$.}

\noindent {\em Proof.}
The reduction is from {\sc DASH}. Let $(D,P)$ be an instance of \textsc{DASH}.  We begin by showing that we may assume, for each vertex $u$ in $P$, we have $d_P^-(u)+d_P^+(u)=1$. To see this, let $D'$ and $P'$ be the digraphs obtained from $D$ and $P$, respectively, by iteratively doing the following for each vertex $v$ in $P$:
\begin{itemize}
\item[(i)] Let $\{s_1,s_2,\ldots,s_i\}$ be the set of in-neighbours of $v$ in $P$ and let
$\{t_1,t_2,\ldots,t_j\}$ be the set of out-neighbours of $v$ in $P$.

\item[(ii)] In $P$, replace $v$ and the arcs $(s_1,v),\ldots,(s_i,v)$ and $(v,t_1),\ldots,(v,t_j)$ with the new vertices $v_1,v_2,\ldots,v_{i+j}$ and the new arcs $(s_1,v_1),\ldots,(s_i,v_i)$ and $(v_{i+1},t_1),\ldots,(v_{i+j},t_j)$.

\item[(iii)] Let $\{x_1,x_2,\ldots,x_r\}$ be the set of in-neighbours of $v$ in $D$ and let
$\{y_1,y_2,\ldots,y_s\}$ be the set of out-neighbours of $v$ in $D$.

\item[(iv)] In $D$, replace $v$ and the arcs $(x_1,v),\ldots,(x_r,v)$ and $(v,y_1),\ldots,(v,y_s)$ with the new vertices $v_1,v_2,\ldots,v_{i+j}$ and the new arcs
\begin{align*}
(x_1,v_1),(x_2,v_1),\ldots, (x_r,v_1),  (x_1,v_2), & (x_2,v_2),\ldots, (x_r,v_2), \\
& \ldots,(x_1,v_i),(x_2,v_i),\ldots,(x_r,v_i)
\end{align*}
and
\begin{align*}
(v_{i+1},y_1),(v_{i+1},y_2),\ldots,  (v_{i+1},y_s), & (v_{i+2},y_1), (v_{i+2},y_2),\ldots, (v_{i+2},y_s),\\
& \ldots,(v_{i+j},y_1),(v_{i+j},y_2),
\ldots,(v_{i+j},y_s).
\end{align*}
\end{itemize}
At the end of this iterative construction, for each vertex $u$ in $P'$, we
have $d_{P'}^-(u)+d_{P'}^+(u)=1$.
Moreover, it is straightforward to check that $P'$ is homeomorphic to a subgraph of $D'$ if and
only if $P$ is homeomorphic to a subgraph of $D$. It now follows that we may assume that our given instance $(D,P)$ of {\sc DASH} is of the form at the completion of this construction.

We next describe a polynomial-time transformation of our instance $(D,P)$ of DASH into an instance   of {\sc Gene Labelling} with $k=1$. Set $k=1$. We define $N$, $\mathcal X$, $\mathcal G$, and the function $G:{\mathcal X}\rightarrow 2^{\mathcal G}$ iteratively as follows. Initially, set $\mathcal X$ and $\mathcal G$ to be both empty. Let $N$ be the phylogenetic network obtained from $D=(V,A)$ by applying the following sequence of operations:
\begin{itemize}
\item[(O-I)]  For each arc $a=(u,v)$ of $P$, add a new gene $g_a$ to $\mathcal G$, add new leaf
vertices $\ell_u,\ell_v$ to $V$ and to $\mathcal X$, add new arcs $(u,\ell_u)$ and $(v,\ell_v)$ to $A$, and set $G(\ell_u)=G(\ell_v)=\{g_a\}$. Furthermore, delete all incoming arcs of $u$ from $A$. At the end of (I), the constructions of the sets $\mathcal X$ and $\mathcal G$, and the function
$G:{\mathcal X}\rightarrow 2^{\mathcal G}$ are completed.

\item[(O-II)] Repeatedly remove all leaves of the resulting network not in
$\mathcal X$ and repeatedly remove all vertices of indegree zero that do not have an element of $\mathcal X$ as a child.

\item[(O-III)] Finally, root the resulting network by choosing a vertex of indegree zero as a root and then adding an arc from this root to each other vertex of indegree zero. Setting $N$ to be the resulting phylogenetic network on $\mathcal X$, we have now constructed the desired instance
of {\sc Gene Labelling}.
\end{itemize}
An example of this construction is shown in Fig.~\ref{fig:homeomorphism}. Note that, while $D$ may not be connected, $N$ is connected because of (O-III). We complete the proof by showing that $N$ admits a
$(G,1)$-gene labelling if and only if $P$ is homeomorphic to a subgraph of $D$.

\begin{figure}[ht]\centering
\includegraphics[width=.8\textwidth]{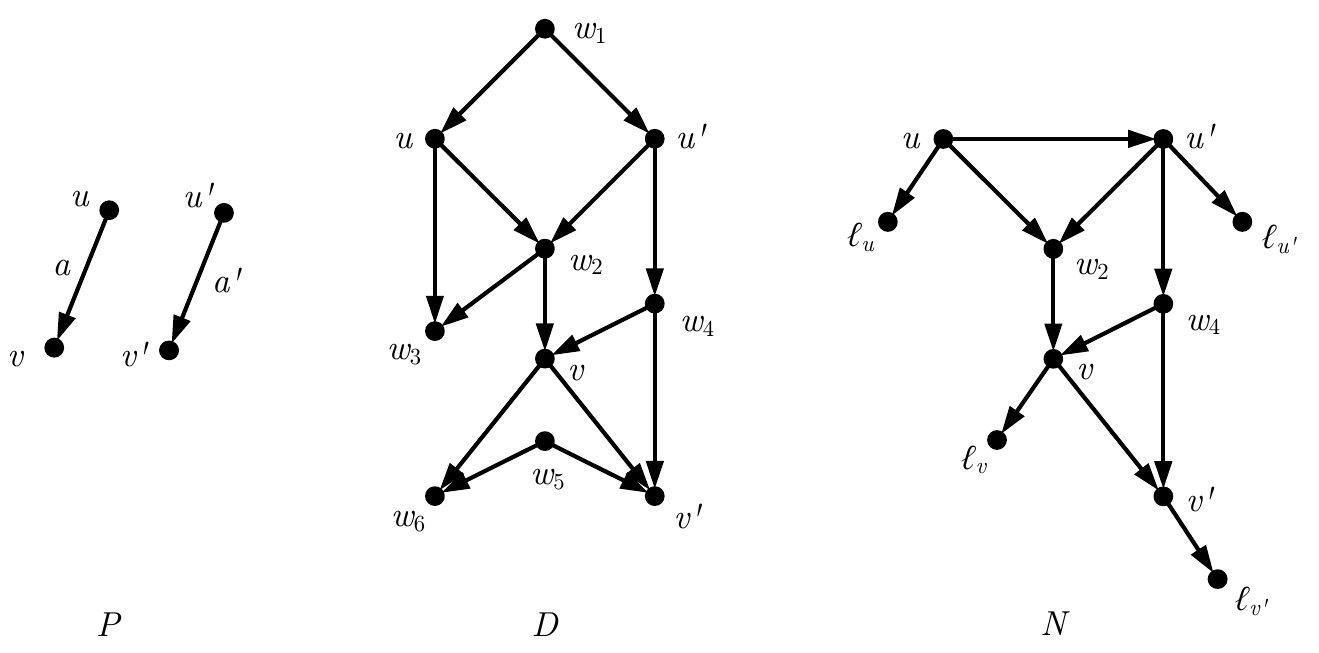}
\caption{An example of the reduction in the proof of Theorem~\ref{prop:homeomorphism}. From an instance~$(P,D)$ of \textsc{DASH}, a phylogenetic network~$N$ is constructed with leaf-labelling~$G(\ell_u)=G(\ell_v)=\{g_a\}$ and~$G(\ell_{u'})=G(\ell_{v'})=\{g_{a'}\}$. Disjoint paths~$u\rightarrow w_2\rightarrow v$ and~$u'\rightarrow w_4\rightarrow v'$ in~$D$ correspond to a labelling~$F(\ell_u)=F(u)=F(w_2)=F(v)=F(\ell_v)=\{g_a\},F(\ell_{u'})=F(u')=F(w_4)=F(v')=F(\ell_{v'})=\{g_{a'}\}$.}\label{fig:homeomorphism}
\end{figure}

Suppose that~$P$ is homeomorphic to a subgraph of~$D$. Then, for each arc~$a=(u,v)$
of~$P$, there exists a directed~$u-v$ path in~$D$ such that all these directed paths are pairwise vertex disjoint. We first claim that for each such~$u-v$ path in~$D$, there exists a
corresponding~$u-v$ path in~$N$. To see this, observe that, in the construction of~$N$ from~$D$, the only arcs deleted are those arcs directed into a vertex, $u$ say, for which $u$ is a vertex in $P$, and arcs incident with a vertex, $w$ say, for which either there is no directed path from $w$ to a vertex in $\mathcal X$ or there is no directed path from a parent of a vertex in $\mathcal X$ to $w$.

None of these deletions deletes an arc on any $u-v$ path in $D$ and so the claim holds.
Now, for each arc~$a=(u,v)$ of~$P$ and for each vertex~$w$ on the associated~$u-v$ path in~$N$, set~$F(w)=\{g_a\}$. Since the children~$\ell_u$ of~$u$ and~$\ell_v$ of~$v$ are the only other vertices with a label containing~$g_a$, the subgraph of~$N=(V,A)$ induced by~$\{w\in V | g_a\in F(w)\}$ is rooted and connected. Labelling all remaining vertices~$w$ by~$F(w)=\emptyset$ thus leads
to a~$(G,1)$-gene labelling of~$N$.

Now suppose that~$F$ is a~$(G,1)$-gene labelling of~$N$. It remains to show that~$P$ is homeomorphic to a subgraph of~$D$. Consider a gene~$g_a\in\mathcal{G}$, and let $a=(u,v)$ be the associated arc of~$P$. Since~$F$ is a~$(G,1)$-gene labelling, the subgraph of~$N$ induced by
$\{w\in V(N): g_a\in F(w)\}$ is connected. Furthermore, each of the arcs added in (O-III) in the construction of $N$ joins two vertices that  are assigned distinct genes in $\mathcal G$ by  $F$ as
$F$ is a $(G,1)$-labelling of $N$. Thus none of these arcs are contained in the subgraph of~$N$ induced by~$\{w\in V(N): g_a\in F(w)\}$. Since~$u$ has no other incoming arcs,~$u$ has indegree zero in this subgraph. Since the child~$\ell_v$ of~$v$ is also labelled~$F(\ell_v)=\{g_a\}$, it follows that~$N$ contains a directed path from~$u$ to~$v$ whose vertices are assigned $\{g_a\}$ under $F$. This path is also a directed path in $D$. Moreover, for two distinct genes $g_a,g_b\in \mathcal G$, these paths are pairwise disjoints and so they are pairwise disjoint in $D$. The union of these paths in $D$ forms a subgraph~$H$ of~$D$ such that~$P$ is homeomorphic to~$H$. This completes the proof of the theorem. \qed

We turn now to the proof of Theorem~\ref{Gktree}, which is based on the concepts of tree-width and tree-decomposition from graph theory -- we define these notions now; for further background the interested reader may wish to consult~\cite{die}.

A \emph{tree decomposition} of a graph~$H=(V_H,E_H)$ is a pair~$(T,\{X_i : i\in I\})$ where~$T=(I,E_T)$ is a tree and, for all $i\in I$, the set $X_i$ is a subset of $V_H$ such that:
\begin{itemize}
\item[(i)] $\bigcup_{i\in I} X_i = V_H$;
\item[(ii)] for each~$(u,v)\in E_H$, there exists an~$i\in I$ with~$u,v\in X_i$;
\item[(iii)]  for each~$v\in V_H$, the subgraph of~$T$ induced by~$\{i\in I : v\in X_i\}$ is connected.
\end{itemize}
The \emph{width} of the tree decomposition is defined as $\max_{i\in I} |X_i|-1$.

We use the following NP-complete problem for the reduction in the proof of the theorem.

\noindent\begin{tabular}{lp{0.85\textwidth}}
\multicolumn{2}{l}{\textsc{Treewidth}}\\
\textit{Given:} & An undirected graph~$H=(V_H,E_H)$ and a natural number~$k'$. \\
\textit{Question:} & Does there exist a tree decomposition of~$H$ with width at most~$k'$? \\
\end{tabular}

\noindent {\bf Theorem~\ref{Gktree}.}
{\em The decision problem {\sc $(G,k)$-Tree} is NP-complete.}

\noindent {\em Proof.}
The reduction is from {\sc Treewidth}. Let $(H,k')$ be an instance of \textsc{Treewidth}, and
set $\mathcal{X}=E_H$, $\mathcal{G}=V_H$, $G(x)=\{u,v\}$ for each edge $x=\{u,v\}\in E_H$, and
$k=k'+1$. We complete the proof by showing that there exists a tree decomposition of~$H$ with
width at most~$k'$ if and only if there exists a phylogenetic tree~$N$ on~$\mathcal{X}$ that admits
a~$(G,k)$-gene labelling.

Firstly, let~$(T,\{X_i : i\in I\})$ be a tree decomposition of~$H$ with width~$k'$. For each~$\{u,v\}\in E_H$, there exists an~$i\in I$ with~$u,v\in X_i$. Hence, for each taxon~${x\in\mathcal{X}}$, there exists a vertex~$i$ of~$T$ with~$G(x)\subseteq X_i$. We construct~$N$ from~$T$ by choosing an arbitrary vertex as a root, directing all edges away from the root and, for each~$x\in\mathcal{X}$, adding a leaf~$x$ and an arc~$(i,x)$ where $i$ is an arbitrary vertex of~$T$ with
$G(x)\subseteq X_i$. Repeatedly deleting leaves not in $\mathcal X$, set $N$ to be the resulting rooted phylogenetic tree on $\mathcal X$. We can now obtain a~$(G,k)$-gene labelling~$F$ of~$N$ by setting~$F(x)=G(x)$ for each leaf~$x\in\mathcal{X}$ and~${F(i)=X_i}$ for each other vertex. For each gene~$g\in\mathcal{G}$, the subgraph of~${N=(V,A)}$ induced by~$\{v\in V : g\in F(v)\}$ is connected by property (iii) of a tree decomposition, and is rooted as~$N$ is a rooted phylogenetic tree.

Now suppose that there exists a phylogenetic tree~$N$ on~$\mathcal{X}$ and a~$(G,k)$-gene labelling~$F$ of~$N=(V,A)$. Then we can obtain a tree decomposition ${(T,\{X_i : i\in I\})}$ of~$H$ by setting~$I=V$ and $X_i=F(i)$ for all $i\in I$, and defining~$T$ to be the tree obtained from $N$ by ignoring the rooting and thus orientation of each of the arcs. All properties of a tree decomposition are clearly satisfied, and the width is at most~$k'=k-1$ because~$|F(i)|\leq k$ by the definition of
a~$(G,k)$-gene labelling.
\qed


\section{$\ldots$ BUT SOMETIMES IT IS EASY}
\label{easysec}

Let $N$ be a galled tree on $\cX$, let $\cG$ be a set of genes, let $G:\cX \rightarrow 2^\cG$ be a genome assignment, and let $k$ be a positive integer.  The main result of this section shows that there is a polynomial-time algorithm for deciding whether $N$ exhibits a $(G,k)$-labelling. If $N$ is a phylogenetic tree, then this problem is equivalent to deciding if $\ell(N,G,k)=0$.

\begin{proposition}
Let $T$ be a phylogenetic tree on $\cX$, let $\cG$ be a set of genes, let $G:\cX \rightarrow 2^\cG$ be a genome assignment, and  let $k$ be a positive integer. Then there is a polynomial-time algorithm for deciding whether $\ell(T,G,k)=0$.
\label{zero}
\end{proposition}

\begin{proof}
Deciding whether $\ell(T,G,k)=0$ is equivalent to deciding if $T$ has a $(G,k)$-gene labelling. With this in mind, it is easily seen that the following $G$-gene labelling function $F$ of $T$ minimizes $k$. For all
$v\in V$, the gene $g\in \cG$ is in $F(v)$ precisely if $v$ is a vertex of the minimal subtree of $T$ that connects those leaves $x$ for which $g\in G(x)$. If $|F(v)|\le k$ for $v$, then $F$ is a $(G,k)$-gene labelling; otherwise there is no such gene labelling of $T$. \qed
\end{proof}

Proposition~\ref{one} (below) establishes the main result when $N$ has exactly one gall. We will use this proposition as the base case for an inductive proof of the main result. The proof of this proposition relies on the following construction. Let $N$ be a galled tree on $\cX$ with exactly one gall. Thus the undirected graph
underlying $N$ has exactly one cycle. Label (in order) the vertices of this cycle $w_1,w_2,\ldots, w_p$, where $w_p$ is the unique vertex in $N$ with two arcs directed into it.

Let $F^*$ be the following map from the vertex set $V$ of $N$ to $2^\cG$. For each  $v\in V$, the gene $g\in \cG$  is in
$F^*(v)$ precisely if, ignoring the direction of the arcs, either:
\begin{itemize}
\item[(i)] there is a pair of leaves $x_1$ and $x_2$ with $g\in G(x_1)$ and $g\in G(x_2)$, and $v$ is on a path between $x_1$ and $x_2$ that avoids $w_p$, or

\item[(ii)] there is a pair of leaves $x_1$ and $x_2$ with $g\in G(x_1)$ and $g\in G(x_2)$, and $v$ is on {\em all} paths between $x_1$ and $x_2$.
\end{itemize}
The following two observations are important for what follows. First, if $F$ is a $G$-gene labelling of $N$, then it is easily seen that $F^*(v)\subseteq F(v)$ for all $v\in V$. Second, $F^*$ is not necessarily a $G$-gene labelling of $N$. The exact reason for this is that there can be a gene $g\in \cG$ such that the sub-digraph of $N$ induced by $\{v\in V : g\in F^*(v)\}$ consists of two rooted connected components; one lying below $w_p$ (more precisely, in the subgraph of $N$ induced by the vertices that are reachable from $w_p$ by a directed path) and at one lying above $w_p$ (more precisely, in the subgraph of $N$ induced by the vertices that are not reachable from $w_p$ by a directed path).

Now let $\cG'$ be the subset of genes $g\in \cG$ for which the sub-digraph of $N$ induced by
$\{v\in V : g\in F^*(v)\}$ is disconnected. We extend $F^*$ to a $G$-gene labelling $F$ of $N$ by reformulating the problem as an undirected network flow problem and then using its solution to identify the extension. Here one can view each edge $\{a,b\}$ as the two arcs $(a,b)$ and $(b,a)$.
We construct an undirected graph $U$ from $N$ by starting with the sub-digraph of $N$  induced by $\{w_1,w_2,\ldots,w_p\}$ and ignoring the direction of the arcs, adding a source vertex $s$, and, for each gene~$g\in\cG'$, adding a new vertex $s_g$ and the three edges $\{s,s_g\}$,
$\{s_g,w_{i_1-1}\}$, and $\{s_g,w_{i_2+1}\}$, where $i_1$ and $i_2$ are the smallest and
largest index $i\neq p$ for which $g\in F^*(w_i)$. Now assign each $s_g$ capacity $1$ and, for each
$i\in \{1,2,\ldots,p-1\}$, assign $w_i$ capacity $k-|F^*(w_i)|$.

To  illustrate the above construction, consider the galled tree $N$ shown in Fig.~\ref{fig:exampleinput}. Each leaf $x$ of $N$ is labelled by the set $G(x)$ of input genes observed in the corresponding taxon. The map $F^*$ is shown in Fig.~\ref{fig:fstar}. The undirected graph $U$ with $k=3$ is shown in Fig.~\ref{fig:auxgraph}.

\begin{figure}[ht]
  \centering
  \begin{subfigure}[]
  {
    \centering
    \includegraphics[scale=0.7]{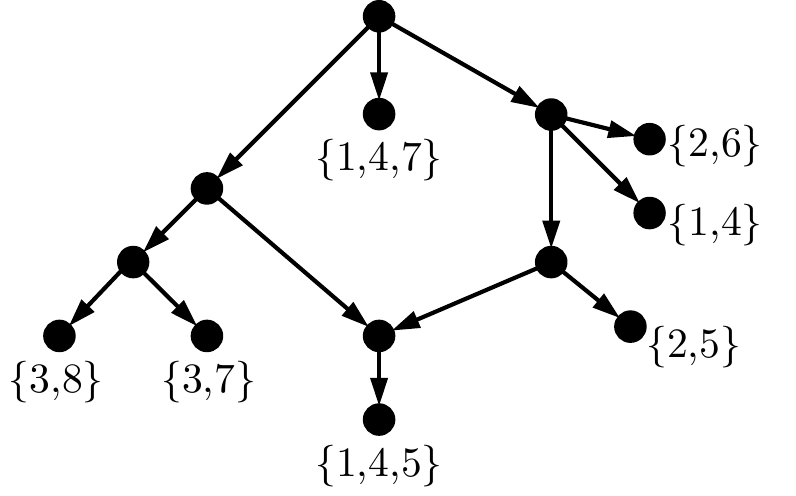}
    \label{fig:exampleinput}
  }
  \end{subfigure}
  \hspace{1cm}
  \begin{subfigure}[]
  {
    \centering
    \includegraphics[scale=0.7]{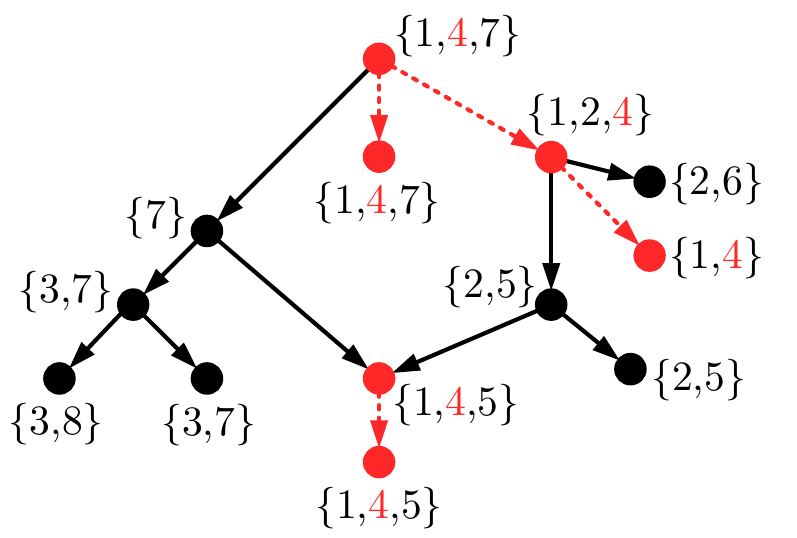}
    \label{fig:fstar}
  }
  \end{subfigure}
    \begin{subfigure}[]
  {
    \centering
    \includegraphics[scale=0.7]{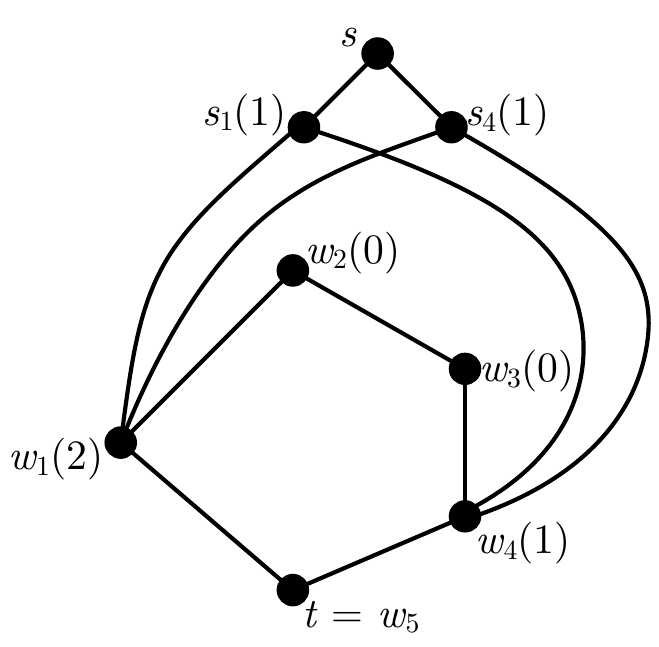}
    \label{fig:auxgraph}
  }
  \end{subfigure}
  \hspace{1cm}
  \begin{subfigure}[]
  {
    \centering
    \includegraphics[scale=0.7]{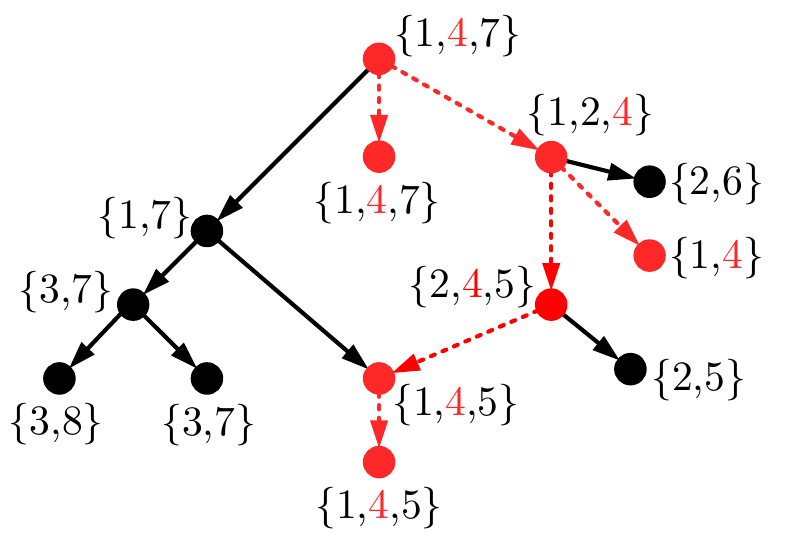}
    \label{fig:fopt}
  }
  \end{subfigure}
  \caption{(a) A galled tree $N$ with one gall. Each leaf $x$ of $N$ is labelled
  by $G(x)$. (b) The initial labelling~$F^*$ in which, for example, the sub-digraph of $N$ induced by
  $\{v\in V : 4\in F^*(v)\}$ (displayed by the dashed arcs and their end vertices) consists of two connected components. (c) Auxiliary graph $U$ with capacities in parentheses. (d) A $(G,3)$-gene labelling of $N$. This gene labelling corresponds to a maximum flow in $U$ which sends one unit of flow through $s_1$ and $w_1$ and one unit of flow through $s_4$ and $w_4$.
}
  \label{fig:example}
\end{figure}

\begin{lemma}
There exists an integer flow in $U$ from $s$ to $w_p$ with value $|\cG'|$ if and only if there exists a
$(G,k)$-gene labelling of $N$.  Moreover, if there is such an integer flow, then it leads to a
$(G,k)$-gene labelling of $N$.
\label{flow}
\end{lemma}

\begin{proof}
First suppose that there exists such a flow $f$ with value $|\cG'|$.
Based on $f$, we show that there exists a $(G,k)$-labelling $F$ of $N$. For this existence proof, we assume that we know the path that each unit of flow takes. We will conclude the proof by showing how an actual $(G,k)$-labelling can be constructed.

Initially set $F=F^*$. Since $f$ has value $|\cG'|$ and each $s_g$ has capacity $1$, there is exactly one unit of flow passing through $s_g$ from $s$ to $w_p$. Furthermore, as $f$ is integer, it uses exactly one of the two edges $\{s_g,w_{i_1-1}\}$ and $\{s_g,w_{i_2+1}\}$. If $f$ uses $\{s_g,w_{i_1-1}\}$, then the corresponding unit of flow either uses  the vertices on the path from $w_{i_1-1}$ to $w_p$ through $\{w_1,w_p\}$ or the vertices on the path from $w_{i_1-1}$ to $w_p$ through $\{w_{p-1},w_p\}$. Depending on which of these paths this unit of flow takes, add $g$ to $F(w_i)$ for each of the vertices on this path. Similarly, if $f$ uses $\{s_g,w_{i_2+1}\}$, then the corresponding unit of flow either uses  the vertices on the path from $w_{i_2+1}$ to $w_p$ through $\{w_1,w_p\}$ or the vertices on the path from $w_{i_2+1}$ to $w_p$ through $\{w_{p-1},w_p\}$. Depending on which of these paths this unit of flow takes, add $g$ to $F(w_i)$ for each of the vertices on this path. Doing this for each $g\in \cG'$, we claim that the resulting map $F:V\rightarrow 2^{\cG}$ is a $(G,k)$-labelling of $N$. Clearly, $F$ satisfies (III). Furthermore, as each vertex $w_i$ has capacity $k-|F^*(w_i)|$, the cardinality of $F(w_i)$ is at most $k$. Thus $F$ satisfies (II). It now follows that $F$ is a $(G,k)$-labelling of $N$.

Now suppose that there exists a~$(G,k)$-gene labelling $F$ of $N$. By one of the two observations earlier, $F^*(v)\subseteq F(v)$ for all $v\in V$. Consider a gene $g\in \cG'$. The sub-digraph of $N$ induced by $\{v\in V : g\in F^*(v)\}$ consists of two rooted connected components. However, by (III), the sub-digraph of $N$ induced by $\{v\in V : g\in F(v)\}$ is rooted and connected. Therefore, there is a path on the cycle consisting of vertices $w_i$ with $g\in F(w_i)-F^*(w_i)$ that connects the two components. Sending one unit of flow from $s$ to the first vertex on this path via $s_g$, and then along this path to $w_p$ for each $g\in \cG'$ gives a desired integer flow.

We have now shown that if there is an integer flow $f$ from $s$ to $w_p$ with value $|\cG'|$, then there is a $(G,k)$-labelling of $N$. This does not directly give such a labelling as we can make no distinction on the flow units. In particular, it is not directly clear which of the two paths a flow unit takes once it reaches a vertex $w_i$ in the cycle. This can be rectified as follows. Let $f$ be such a flow and let $g\in \cG'$. Ignoring the vertices $w_{i_1},\ldots,w_{i_2}$, either the flow unit through $s_g$ takes the path from $w_{i_1-1}$ to $w_p$ via $w_1$ or the path from $w_{i_2+1}$ to $w_p$ via $w_{p-1}$. To make this decision, consider the following modification of the integer flow problem. Extend $F^*$ to $F^*_g$ by adding $g$ to each of $F^*(w_{i_1-1}),\ldots,F^*(w_1)$ and, for each of these vertices, subtract one from their capacities. If there is an integer flow from $s$ to $w_p$ in $U\ba s_g$ of $|\cG'|-1$ units, then we may assume that the unit of flow through $s_g$ in $U$ follows the path from $w_{i_1-1}$ to $w_p$ via $w_{p-1}$. In this case, replace $F^*$ with $F^*_g$ and $U$ with $U\ba s_g$, and repeat for another element in $\cG'-g$. If there is no such integer flow in $U\ba s_g$, then the unit of flow through $s_g$ in $U$ follows the path from $w_{i_2+1}$ to $w_p$ via $w_{p-1}$. In this second case, replace $F^*$ with that obtained by adding $g$ to each of $F^*(w_{i_1-1}),\ldots,F^*(w_1)$ and, for each of these vertices, subtract one from their capacities, and replace $U$ with $U\ba s_g$. Continuing in this way, we eventually obtain a $(G,k)$-labelling of $N$. \qed
\end{proof}

To illustrate Lemma~\ref{flow} and its proof, consider the example prior to the lemma, illustrated in Fig.~\ref{fig:example}. In $U$, a maximum flow could send either two units of flow through $w_1$ or one unit of flow through vertex $w_1$ and one unit of flow through vertex $w_4$. From the latter option, one can for example obtain the $(G,3)$-gene labelling shown in Fig.~\ref{fig:fopt}.

\begin{proposition}
Let $N$ be a phylogenetic network on $\cX$, let $\cG$ be a set of genes, let $G:\cX \rightarrow 2^\cG$ be a genome assignment, and let $k$ be a positive integer.
\begin{itemize}
\item[{\rm (i)}] If $N$ is a galled tree with exactly one gall, then there is a polynomial-time algorithm for deciding whether $N$ exhibits a $(G,k)$-labelling, in which case, such a labelling can also be found in polynomial time.

\item[{\rm (ii)}] If $T$ is a phylogenetic tree, then there is a polynomial-time algorithm for deciding whether $\ell(T,G,k)=1$.
\end{itemize}
\label{one}
\end{proposition}

\begin{proof}
First note that a maximum-valued integer flow can be found in $O(n^{1.5}\log(n\cdot k))$ time~\cite{GoldbergRao1998}. Thus, (i) follows from Lemma~\ref{flow}. For (ii), if $|\cX|=n$, then there is $O(n^2)$ possible ways of adding a single arc to $T$. Applying Lemma~\ref{flow} to each such way gives the desired algorithm. This completes the proof of the proposition. \qed
\end{proof}

We now extend Proposition~\ref{one}(i) to all galled trees using induction on the number of galls.  Let $N$ be a galled tree on $\cX$, let $\cG$ be a set of genes, let $G:\cX \rightarrow 2^\cG$ be a genome assignment and let $k$ be a positive integer. If $N$ has either no galls or exactly one gall, then we have such an algorithm by Propositions~\ref{zero} and~\ref{one}, so we may assume that $N$ has at least two galls. In this case there exists a vertex $u_1$ of $N$ with the property that, for some gall, each of the vertices in the vertex set of this gall are descendants of $u_1$ and no vertex that is a proper descendant $u_1$ has this property.  Let $N_1$ be the phylogenetic network obtained from $N$ by replacing $u_1$ and all of its descendants with a single vertex $q_1$. Let $Q_1$ be the phylogenetic network obtained from $N$ by deleting all of the vertices of $N$ that are not descendants of $u_1$ and adjoining a parent vertex $r_1$ to $u_1$ with one further child other than $u_1$. Call the additional child vertex $v_1$. Let $L_{Q_1}$ denote the leaf set of $Q_1$. Effectively, we have partitioned $N$ into two phylogenetic networks $N_1$ and $Q_1$. See Figure~\ref{fig:galledtree} for an example.
Let
$$G(q_1)= G(v_1)=\big(\bigcup_{x\in L_{Q_1}-\{v_1\}} G(x)\big)\cap
\big(\bigcup_{x\in \cX-L_{Q_1}}G(x)\big).$$
The proof of the following lemma is straightforward, and so the details are omitted.

\begin{figure}[ht]\centering
\includegraphics[width=11cm]{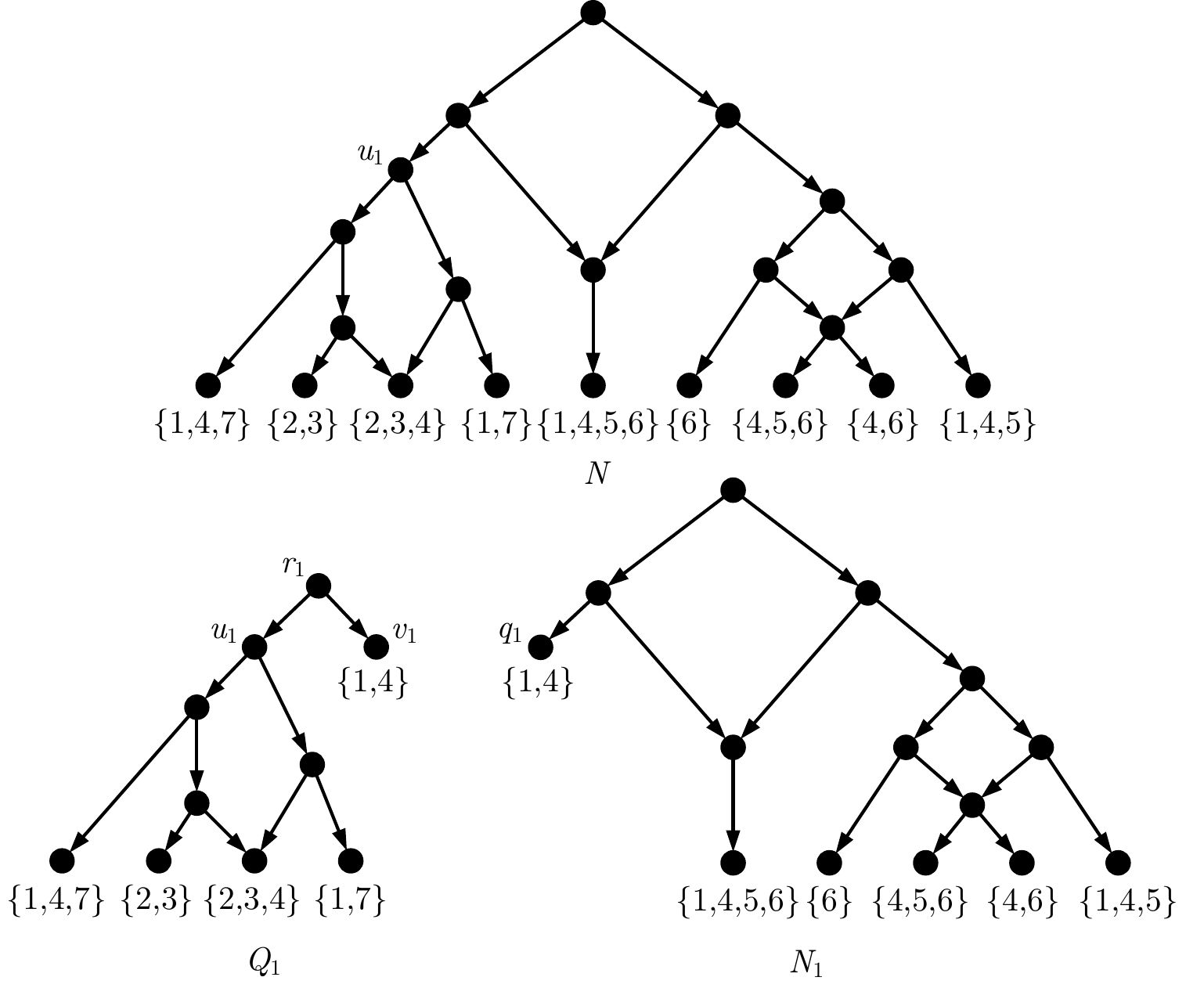}
\caption{A galled tree~$N$ and the decomposition of~$N$ into~$Q_1$ and~$N_1$ described in the text.}\label{fig:galledtree}
\end{figure}

\begin{lemma}
The galled tree $N$ has a $(G,k)$-labelling if and only if each of $N_1$ and $Q_1$ has a
$(G,k)$-labelling.
\label{repeat}
\end{lemma}

By Proposition~\ref{one}(i), there is a polynomial-time algorithm for deciding whether or not $Q_1$ has a $(G,k)$-labelling. If there is no such labelling, then, by Lemma~\ref{repeat}, $N$ has no $(G,k)$-labelling. On the other hand, if $Q_1$ has a $(G,k)$-labelling, then one needs to check if $N_1$ has a $(G,k)$-labelling. Now repeat the above construction with $N$ replaced by $N_1$. Continuing in this way, we either find a galled tree with a single gall that does not exhibit a $(G,k)$-labelling, and thereby show  that $N$ has no such labelling, or we find no such galled tree and conclude that $N$ has a
$(G,k)$-labelling. Note that the number of galls in $N$ is polynomial in the size of the vertex set of $N$.  In particular, we have established the following results.

\medskip
\noindent {\bf Theorem~\ref{prop:level1}}
{\em Let $N$ be a galled tree on $\cX$, let $\cG$ be a set of genes, let $G:\cX \rightarrow 2^\cG$ be a genome assignment, and let $k$ be a positive integer. Then there is a polynomial-time algorithm for deciding whether $N$ exhibits a $(G,k)$-gene labelling.}

\begin{corollary}
 \label{gallcor}
Let $T$ be a rooted phylogenetic tree on $\cX$, let $\cG$ be a set of genes, let $G:\cX \rightarrow 2^\cG$ be a genome assignment, and let $k$ be a positive integer. If $h$ is a fixed non-negative integer, then there is a polynomial-time algorithm for deciding whether or not there is a galled tree $N$ on $\cX$ that can be obtained from $T$ by adding at most $h$ arcs and which exhibits a $(G,k)$-gene labelling.
\end{corollary}

\begin{proof}
Suppose that~$N$ is a galled tree on~$\cX$ that can be obtained from~$T$ by adding at most~$h$ arcs. Then there is an embedding~$T'$ of~$T$ in~$N$. Notice that since~$N$ is a galled tree, it follows that all vertices of~$N$ are contained in~$T$ and thus that~$N$ can be obtained from~$T$ by subdividing at most~$2h$ arcs and adding at most~$h$ arcs.

Hence, given~$T$, we can try each possible way of subdividing at most~$2h$ arcs and adding at most~$h$ arcs. For each such possibility, we check if the resulting network is a galled tree.  In each such case we can check if a~$(G,k)$-gene labelling of this network exists, by Theorem~\ref{prop:level1}. The time needed is polynomial in the size of the input, for each fixed~$h$. \qed
\end{proof}

\section{CONCLUDING COMMENTS}

The analysis of this paper rests on a number of assumptions concerning gene evolution. Perhaps the most restrictive is the requirement that gene genesis is a unique event.  This requirement reflects the fact that a gene is typically a long and fairly precise sequence of nucleotides, and the probability that a similar sequence could evolve independently in a different part of the tree is small. This seems reasonable if DNA sequence evolution is described by a neutral model \cite{kim},  but, in some cases, natural selection will, no doubt, direct the evolution of DNA sequences towards certain genes that confer higher fitness. Thus,  simple arguments based on neutrality need to be treated with caution. It would be interesting to extend the analysis of this paper to allow for a small frequency of independent gene genesis events.

A related question is what degree of sequence similarity is required in order to classify two sequences as coding for the same gene.  Insisting on exact sequence identity is too severe, since it is well known that different species typically  encode a gene with slightly different sequences that result from random site substitutions (indeed these differences have been the main signal used for phylogenetic tree reconstruction \cite{fels}).  This question of gene identity is also relevant to the probability of independent gene genesis:  a region of DNA that codes for a gene could, in principle, accumulate sufficient site mutations to put it just outside the range of being identified with that gene, but could  then mutate back within range, giving the appearance of a second gene genesis event.

Other aspects of the model that may be criticized are the assumptions that the species tree is known with certainty (or, indeed, that it is meaningful to talk of a `species tree' \cite{doo07}), and that the model does not penalize gene losses at all.

Our computational complexity results  highlight that many problems are surprisingly difficult, even for a tree,  and  some questions still remain to be explored further. One that seems particularly interesting is
described as follows, along with our conjecture as to its possible resolution.

Given a rooted phylogenetic tree~$T$, a set of genes~$G(x)$ for each leaf~$x$ of~$T$, and natural numbers~$k$ and~$h$, consider the problem of deciding whether it is possible to add at most~$h$ arcs to~$T$ to obtain a phylogenetic network~$N$ that admits a~$(G,k)$-gene labelling.

\begin{conjecture} \label{con:networkfromtree}
This problem is NP-hard in general, but for each fixed $h$, it admits a polynomial-time algorithm.
\end{conjecture}

\end{document}